\documentclass[runningheads]{llncs}
\usepackage[T1]{fontenc}
\usepackage{graphicx}
\usepackage{hyperref}
\usepackage{amsfonts}
\usepackage{mathtools}
\usepackage{xcolor}
\usepackage{bbm}
\usepackage{tikz}
\usepackage{eurosym}
\usepackage{comment}
\usepackage{xcolor}
\usepackage{wrapfig}
\usepackage{lineno}

\usepackage[linesnumbered,lined,boxed,commentsnumbered,ruled]{algorithm2e}

\newcommand{\A}{\mathcal{A}}
\newcommand{\B}{\mathcal{B}}
\renewcommand{\b}{\mathbf{b}}
\renewcommand{\v}{\mathbf{v}}
\newcommand{\q}{\mathbf{q}}
\newcommand{\s}{\mathbf{s}}
\renewcommand{\t}{\mathbf{t}}

\newcommand{\bi}[1]{\widehat{#1}}
\newcommand{\w}{\mathbf{w}}
\renewcommand{\v}{\mathbf{v}}
\newcommand{\set}[1]{\left\{ #1 \right\}}
\newcommand{\tuple}[1]{\left( #1 \right)}
\newcommand{\linCE}{\mathit{NC}}
\newcommand{\logCE}{\mathit{KC}}
\newcommand{\taxCE}{\mathit{TC}}
\newcommand{\triw}{\mathbf{1}}

\DeclareMathOperator*{\argmax}{arg\,max}

\begin{document}
\title{Condorcet Markets}
%
%
\author{St\'{e}phane Airiau\inst{1}
\orcidID{0000-0003-4669-7619} 
\and
Nicholas Kees Dupuis \and
Davide Grossi\inst{2,3}
\orcidID{0000-0002-9709-030X}
}
\authorrunning{Airiau, et al.}
%
\institute{
University Paris-Dauphine, France \\
\email{stephane.airiau@dauphine.fr}
\and
University of Groningen, The Netherlands \\
\email{d.grossi@rug.nl}
\and
University of Amsterdam, The Netherlands
}
\maketitle              
\begin{abstract}
The paper studies information markets concerning single events from an epistemic social choice perspective. Within the classical Condorcet error model for collective binary decisions, we establish equivalence results between elections and markets, showing that the alternative that would be selected by weighted majority voting (under specific weighting schemes) corresponds to the alternative with highest price in the equilibrium of the market (under specific assumptions on the market type). This makes it possible in principle to implement specific weighted majority elections, which are known to have superior truth-tracking performance, by means of information markets without needing to elicit voters' competences. 

\end{abstract}
%
%
%



\section{Introduction}

Information, or prediction, markets are markets of all-or-nothing contracts (so-called Arrow securities) that pay one unit of currency if a designated event occurs and nothing otherwise (see \cite{alos-ferrer05asset,beygelzimer2012learning,kets14betting} for models of such markets). Under the view, inspired by \cite{hayek45use}, that markets are good aggregators of the information dispersed among traders, proponents of information markets have argued that equilibrium prices are accurate estimates of the probability of the designated event. Recent research---theoretical and empirical---has probed this interpretation of prices in information markets, finding that equilibrium prices successfully track the traders' average belief about the event, under several models of trader's utilities \cite{pennock_thesis,wolfers04prediction}. 

In this paper we address a closely related, but different question: {\em if a decision maker takes a decision based on the information they extract from the equilibrium price of the market, how accurate would the decision be?} Therefore, rather than relating equilibrium prices to the aggregation of traders' beliefs, we relate them directly to the quality of the decision they would support. We frame the above question within the standard binary choice framework of epistemic social choice, stemming from the Condorcet jury theorem tradition \cite{Condorcet1785,grofman83thirteen,young88condorcet} and the maximum-likelihood estimation approach to voting \cite{grofman83thirteen,conitzer05common,pivato13voting,elkind16rationalizations}. 

\begin{wrapfigure}{r}{5cm}
\vspace{-0.6cm}
\begin{center}
\begin{tikzpicture}[node distance=3cm, auto]
    \node (1) {beliefs};
    \node (2) [below of=1] {votes};
    \node (3) [right of=1] {price};
    \node (4) [below of=3] {decision};
    \draw[->] (1) to node {\rotatebox{-90}{binarize}} (2);
    \draw[->] (3) to node {\rotatebox{-90}{binarize}} (4);
    \draw[->] (1) to node {market} (3);
    \draw[->] (2) to node {election} (4);
\end{tikzpicture}
\end{center}
\caption{Elections and information markets commute.}
\vspace{-0.4cm}
\label{fig:commute}
\end{wrapfigure}
\paragraph{Contribution} To answer the above question, we study information markets when traders' beliefs are obtained by Bayesian update from a private independent signal with accuracy known to the trader, just like in the classic jury theorems setting. In other words, we study `jurors' as if they were `traders' who, instead of relaying their vote to a central mechanism, trade in an information market. In taking this perspective, we compare the decisions that would be taken based on the equilibrium price of an information market, with the decisions that would be taken by specific weighted majority elections, whose truth-tracking behavior is already well-understood \cite{grofman83thirteen}. Specifically, we aim at identifying correspondences between classes of markets and of weighted majority elections which are equivalent from a decision-making point of view. That is, if agents vote according to the event they believe more likely and aggregate these votes by weighted majority, then they identify the alternative whose Arrow security would have highest price in the equilibrium of the market in which the same agents trade based on their beliefs. Figure \ref{fig:commute} depicts such relationship via a commutative diagram. This type of results open up the possibility (in principle) of implementing weighted majority voting with proven truth-tracking performance without needing to know (or estimate \cite{baharad10beyond}) jurors' competences. 

\paragraph{Paper outline}
Section \ref{sec:prelim} introduces the standard binary truth-tracking framework and presents our model of information markets. Section \ref{sec:Naive_Kelly} presents results on equilibrium prices in two of the three types of markets we consider (Naive and Kelly markets) and Section \ref{sec:Naive_Kelly_truth} proves `Figure \ref{fig:commute}-type' results for those markets. Section \ref{sec:perfect} then shows how such results could be lifted even to the case of majority voting where jurors are weighted perfectly according to their competence. 
Section  \ref{sec:examples} discusses two examples illustrating our framework and analysis. 
Section \ref{sec:conc} outlines future research directions. 
Auxiliary results are provided in Appendix \ref{sec:supp} and proofs of known results are provided for completeness in Appendix \ref{app:average}.



\section{Preliminaries} \label{sec:prelim}

\subsection{Collective truth-tracking}

\paragraph{Collective decisions.} We are concerned with a finite set of agents $N = \set{1, \ldots, n}$ who have to decide collectively on the correct state of the world $x \in \{ A, B \}$. A prior probability $P(x = A) = \pi = 0.5$ is given, that the correct state is $A$. Each agent $i$ observes a private independent signal $y_i \in \set{A,B}$ that has quality $q_i \in (0.5,1]$. Each $q_i$ represents the competence or accuracy of $i$, which is assumed to satisfy $q_i = P(y=A \mid x = A) = P(y=B \mid x = B)$. We call each vector $\q = (q_1, \ldots, q_n)$ of individual accuracies an {\em accuracy} or {\em competence profile} of the group. Having observed her private signal, each agent then forms a posterior belief $b_i = P(x = A \mid y = A)$ about state $x=A$ by Bayes rule. By Bayes' rule and the condition on the prior, we have that either $b_i = q_i > 0.5$ (the belief in $A$ equals $q_i$) or $b_i = 1 - q_i < 0.5$ (the belief in $A$ equals $1-q_i$). This gives us, for all $i \in N$:
\begin{equation}
b_i =  \mathbbm{1}(b_i > 0.5) \cdot (2 q_i-1) + (1-q_i), \label{eq:belief}
\end{equation}
where $\mathbbm{1}$ denotes the indicator function. Individual beliefs are  then collected in a {\em belief profile} $\b = (b_1, \ldots, b_n) \in [0,1]^n$. Given an accuracy profile $\q$, the set of possible belief profiles is denoted $\B_\q = \{ \b \in [0,1]^n \mid P(\b \mid \q)>0 \}$. Observe that the size of this set equals $2^n$: the number of all signal realizations. 

\medskip

Based on a profile $\b$ of individual beliefs, the group then takes a decision by mapping the profile to $A$ or to $B$. In this process, agents may have different weights, which are collected in a {\em weight profile} $\w = (w_1, \ldots, w_n) \in [0,1]^n$. We refer to $\triw = (1, \ldots, 1)$ as the {\em egalitarian weight profile} in which all agents have equal weight. Assuming a weight profile $\w$, we call an {\em aggregator} any function 
\begin{equation}
\A^{\w}: [0,1]^n \to 2^{\{1, 0 \}} \backslash \set{\emptyset},
\end{equation}
mapping belief profiles to alternatives, where $\set{1}$ denotes $A$; $\set{0}$ denotes $B$; and $\set{1,0}$ denotes a tie.

\paragraph{Types of aggregators.}
We will study two classes of mechanisms to implement aggregators. In the first class, agents cast binary ballots based on their beliefs and these ballots are submitted to a voting mechanism. The winning alternative is the outcome of the aggregation process. In the second class, agents' trade in special types of securities, based on their beliefs. The equilibrium price of this securities market is then used as a proxy for the group's belief in the probability of state $A$. In this case, it is the alternative favored by this collective belief to be the outcome of the aggregation process. 

Let us make the above notions more precise. First of all, a belief $b \in [0,1]$ is translated into binary opinions, or {\em votes}, for $A$ or $B$ via the binarization function $\bi{}: [0,1] \to 2^{\{1, 0 \}} \backslash \emptyset$ defined as follows:
\begin{align} \label{eq:binarization}
    \bi{b} & = 
    \begin{cases}
    \{ 1 \} & \mbox{if} \ b > 0.5, \\
    \{ 0 \} & \mbox{if} \ b < 0.5, \\
    \{ 0, 1 \} & \mbox{otherwise}.
    \end{cases}
\end{align}
That is, agents are assumed to vote in accordance to their posterior belief (this is sometimes referred to as sincere voting \cite{austen-smith96information}). A binarized belief profile $\bi{\b} = (\bi{b}_1, \ldots, \bi{b}_n)$ is therefore a binary vector and we will referred to such vectors also as {\em voting profiles} and denote them by $\v = (v_1, \ldots, v_n)$.\footnote{As individual beliefs cannot equal $0.5$, the reduction function always outputs a singleton $\{ 0 \}$ or $\{ 1 \}$ on individual beliefs. We will see, however, that this is not the case for collective beliefs.}

Given a weight profile $\w$, a (belief) merger is a function $F^\w: [0,1]^n \to [0,1]$ taking as input a belief profile and outputting a group belief. A choice function is a function $f^\w: \set{1,0}^n \to 2^{\set{0,1}} \backslash \emptyset$ taking as input a voting profile and outputting a possibly tied choice between $1$, i.e., $A$, and $0$, i.e., $B$.  We will study aggregators of the type $f^\w \circ \bi{ \ }$ (voting) and $\bi{ \ } \circ F^\w$ (trading), where $\circ$ denotes function composition. 
A voting mechanism is a choice function $f^w$ which, applied to a binarized belief profile $\bi{\b}$, yields a collective choice $f^\w(\bi{\b})$ (under the weight profile $\w$). A market mechanism is a belief aggregation function $F$ that, once applied to a belief profile $\b$, yields a collective belief $F^\w(\b)$ whose binarization $\bi{F^\w(\b)}$ yields a collective choice (under the weight profile $\w$). 

We are concerned with the truth-tracking performance of aggregators.
The accuracy 
of an aggregator $\A^\w$ under the accuracy profile $\q$, is the conditional probability that the outcome of the aggregator is $x$ if the state of the world is $x$.
What we outlined describes an epistemic social choice setting where the group is confronted with a maximum-likelihood estimation task in a dichotomous choice situation (see \cite{elkind16rationalizations}).


\subsection{Voting and market mechanisms}

We turn now to the description of the mechanisms we are concerned with.

\subsubsection{Voting mechanisms}

After observing their private signal, agents decide whether to vote for $A$ or $B$ according to Equation \eqref{eq:binarization}. A weighted majority rule is then applied to these votes to determine the group's choice:
\begin{align} \label{eq:maj}
    M^\w(\v) & = 
    \begin{cases}
    \{ 1 \} & \mbox{if} \ \sum_{i \in N} w_i v_i > \frac{\sum_{i \in N} w_i}{2}, \\
    \{ 0 \} & \mbox{if} \ \sum_{i \in N} w_i v_i < \frac{\sum_{i \in N} w_i}{2}, \\
    \{ 0, 1 \} & \mbox{otherwise}.
    \end{cases}
\end{align}
We will be working in particular with three variants of Equation \eqref{eq:maj} defined by three different weight profiles: the egalitarian weight profile $\triw$; the weight profile allocating to each agent $i$ a weight proportional to $q_i - 0.5$; the weight profile allocating to each agent $i$ a weight proportional to $\log \frac{q_i}{1-q_i}$. The first weight profile defines the {\em simple majority} rule. The second weight profile simulates decision-making according to the mean belief of the group. The latter weight profile can be inferred from Bayes theorem and induces the weighted majority rule which we refer to as {\em perfect majority}, and which has been proven to optimize the truth-tracking ability of the group.

\begin{theorem}[\cite{grofman83thirteen}] \label{th:grofman}
For any competence profile $\q$, 
the accuracy of $M^\w$ given $\q$
is maximal if $\w$ is such that $w_i \propto \ln \left( \frac{q_i}{1-q_i} \right)$ for all $i \in N$.
\end{theorem}

\subsubsection{Markets}

The market model we use is borrowed from \cite{kets14betting,beygelzimer2012learning}. Two symmetric Arrow securities are traded: securities of type $A$, which cost $p_A \in [0,1]$ and pay $1$ unit of currency if $x = A$, and $0$ otherwise; securities of type $B$, which cost $p_B \in [0,1]$ and pay $1$ unit if $x = B$ and $0$ otherwise.  After observing their private signal, agents decide what fraction of their endowment to invest in which securities. The endowment is fixed to $1$ for all agents.
 When the true state of the world is revealed, the market resolves and payouts based on the agents' investments are distributed. We refer to tuples $\s^A = \tuple{s^A_1, \ldots, s^A_n}$ (respectively, $\s^B = \tuple{s^B_1, \ldots, s^B_n}$) as investment profiles in $A$-securities (respectively, $B$-securities). We refer to a pair $\s = (\s^A,\s^B)$ as an {\em investment profile}.  We assume that agents {\em invest in at most one} of these securities, so if $s^A > 0$ then $s^B = 0$ and vice versa. In our setting, this assumption is without loss of generality (see Proposition \ref{prop:wlog} in the appendix\footnote{We are indebted to Marcus Pivato for bringing this issue to our attention.}). We call agents investing in $A$, $A$-{\em traders} and agents investing in $B$, $B$-{\em traders}. We proceed now to define the notions of price, utility and equilibrium.

\paragraph{Market mechanism}
When the market opens, all purchasing orders for each security are executed by the the market operator. The market operator sells all requested securities to agents when the market opens and pays the winning securities out immediately when the market resolves, that is, when either $A$ or $B$ turns out to be the case.
We further assume that the operator makes no profits and incurs no losses. So, for every $A$-security sold at price $p^A$, a $B$-security is sold at price $p^B = 1-p^A$ and vice versa. It follows that there are as many $A$-securities as $B$-securities and the price of the risk-less asset consisting of one of each security is $p_A + p_B = 1$. In this way the operator finances the payout of any bet by the pay-in of the opposite bet.


Under the above assumptions, the market clears\footnote{A market is said to clear when supply and demand match. In our model, supply and demand are implicit in the following way: demand for an $A$-security at price $p^A$ implies supply for a $B$-security at price $p^B = 1 - p^A$ and vice versa. The same applies to supply and demand for $B$ securities.} when the total amount of individual wealth invested in $A$-securities, divided by the price of $A$-securities (demand of $A$-securities) matches the amount of individual wealth invested in $B$-securities, divided by the price of $B$-securities (demand of $B$-securities):\footnote{It may be worth observing that by the above design we are effectively treating the operator as an extra trader in the market, who holds a risk-less asset consisting of $\frac{1}{p^A}\sum_{i \in N} s^A_i$ $A$-securities and $\frac{1}{1 - p^A} \sum_{i \in N} s^B_i$ $B$-securities. We are indebted to Marcus Pivato for this observation.}
\begin{equation} \label{eq:balance}
\frac{1}{p^A}\sum_{i \in N} s^A_i =  \frac{1}{1 - p^A} \sum_{i \in N} s^B_i. 
\end{equation}
It follows that, given an investment profile $\s$, solving Equation \eqref{eq:balance} for $p^A$, yields the clearing price $\frac{\sum_{i \in N} s^A_i}{\sum_{i \in N} s^A_i + \sum_{i \in N} s^B_i}$, denoted $p^A(\s)$. Note that the price is undefined if $p^A = 0$ or $p^A = 1$. We come back to this issue in Remark \ref{remark:null}.

When the market resolves, each agent receives a different payout depending on how much of each security she owns, how the market resolves, and how much of her endowment is not invested. The payout, that is, the amount of wealth obtained by $i$ with a given strategy $s^A_i$ investing in $A$ under a price $p^A$, is:
\begin{equation} \label{eq:clearing}
        z(p^A, s^A_i) = \left\{
        \begin{array}{ll}
            \frac{s^A_i}{p^A} & \quad A \text{ is correct},\\
            1-s^A_i & \quad \mbox{otherwise},
        \end{array}
    \right.
\end{equation}
where $\frac{s^A_i}{p^A}$ equals the amount of $A$-securities that $i$ has purchased.
The payout for an investment in $B$-securities is defined in the same manner. 

\begin{remark}
For simplicity, in what follows we will refer to the price of $A$-securities as $p$ instead of $p^A$ and to the price of $B$-securities as $1-p$ instead of $p^B$.
\end{remark}

\paragraph{Utility}
We study price $p$ by making assumptions on how much utility agents extract from their payout at that price. We consider two types of utility functions:
\begin{description}
   
\item[Naive] Given a price $p \in [0,1]$, the naive utility function of an $A$-trader $i$ is
$
u(p, s^A_i) = z(p,s^A_i)
$
Similarly, for a $B$-trader, it is 
$
u(1-p, s^B_i) = z(1-p,s^B_i).
$
The expected utility for investment in $A$-securities is then:
\begin{equation} \label{eq:nu}
U^A_i(p,s^A_i) = \mathbb{E}[u(p,s^A_i)] = b_i \left(\frac{s^A_i}{p} - s^A_i + 1 \right) + (1-b_i) (1-s^A_i). 
\end{equation}
The expected utility for investment in $B$-securities is, correspondingly, $b_i (1-s^B_i) + (1-b_i) \left(\frac{s^B_i}{1-p} - s^B_i + 1\right)$.
We will refer to markets under a naive utility assumption as {\em Naive markets}.

\item[Kelly] Given a price $p \in [0,1]$, the Kelly \cite{kelly56new} utility function of an $A$-trader $i$ is $u(p,s^A_i) = \ln(z(p,s^A_i))$, and mutatis mutandis for $B$-traders. The expected Kelly utility for an $A$-trader is therefore:
\begin{equation}
U^A_i(p,s^A_i) = \mathbb{E}[u(p,s^A_i)] = b_i \ln \left(\frac{s^A_i}{p} - s^A_i + 1 \right) + (1-b_i) \ln (1-s^A_i). \label{eq:kelly}
\end{equation} 
Correspondingly, the expected utility of investment $s^B_i$ for a $B$-traders is $b_i \ln(1-s^B_i) + (1-b_i) \ln \left( \frac{s^B_i}{1-p} - s^B_i + 1 \right)$.
We will refer to markets under such logarithmic utility assumption as {\em Kelly markets}. Investing with a logarithmic utility function is known as Kelly betting and is known to maximize bettor's wealth over time \cite{kelly56new}. Information market traders with Kelly utilities have been studied, for instance, in \cite{beygelzimer2012learning}.
\end{description}

\paragraph{Equilibria} For each of the above models of utility we will work with the notion of equilibrium known as competitive equilibrium \cite{pennock_thesis}. This equilibrium assumes that agents optimize the choice of their investment strategy $s_i$ under the balancing assumption of Equation \eqref{eq:balance}, while not considering the effect of their choice on the price (they behave as `price takers'). 

\begin{definition}[Competitive equilibrium]
Given a belief profile $\b$, an investment profile $\s$ is in competitive equlibrium for price $p^\star$ if and only if:
\begin{enumerate}
\item Equation \eqref{eq:balance} holds, that is, $p^\star = p(\s)$,
\item for all $i\in N$, if $i$ is a $t$-trader in $\s$, then $s^t_i \in \argmax_{x \in [0,1]} U_i^t(p^t,x)$, for $t \in \set{A,B}$.
\end{enumerate}
\end{definition}
So, when the investment profile $\s$ is in equilibrium with respect to the $A$-securities $p^\star$, no agent would like to purchase more securities of any type given their beliefs. If $\s$ is in equilibrium for price $p(\s)$, then we say that $\s$ is an {\em equilibrium}. If equilibria always exist, and are such for one same price, then the equilibrium price can be interpreted as the market's belief that the state of the world is $A$, given the agents' underlying beliefs $\b$. We can therefore view a market as a belief merger $F^\w: [0,1]^n \to [0,1]$, mapping beliefs to the equilibrium price.

\begin{remark}[Null price] \label{remark:null}
Under equation \eqref{eq:balance} a price $p = 0$ (respectively, $p=1$) implies that there are no $A$-traders (respectively, no $B$-traders).
In such cases Equations \eqref{eq:clearing}, \eqref{eq:nu} (Naive utility) and \eqref{eq:kelly} (Kelly utility) would be formally undefined. Such situations, however, cannot occur in equilibrium because as $p$ approaches $0$ (respectively, $1$), the utility for $s^A_i > 0$ (respectively, $s^B_i>0$) approaches $\infty$ under both utility models. No investment profile can therefore be in equilibrium with respect to prices $p=0$ or $p=1$.
\end{remark}


\section{Equilibrium price in Naive and Kelly markets} \label{sec:Naive_Kelly}

In order to see markets as belief aggregators we need to show that the above market types always admit equilibria and, ideally, that equilibrium prices are unique, thereby making the aggregator resolute. We do so in this section.

\subsection{
Equilibrium $p$ in Naive markets is the ($1-p$)-quantile belief
}


Let us start by observing that, under naive utility, agents maximize their utility by investing all their wealth, unless their belief equals the price, in which case any level of investment would yield the same utility to them in expectation. 

\begin{lemma} \label{lemma:naive}
In Naive markets, for any competence profile $\q$, belief profile $\b \in \B_\q$, and price $p \in [0,1]$ we have that, for any $i \in N$:
 \begin{align*}
 \argmax_{x \in [0,1]} U_i^A(p,x) & =     
 \begin{cases}
    \{ 1 \} & \mbox{if} \ p < b_i, \\
    \{ 0 \} & \mbox{if} \ p > b_i, \\
    [0,1] & \mbox{otherwise}, 
    \end{cases}  \\
   \argmax_{x \in [0,1]} U_i^B(p,x) & =     
  \begin{cases}
    \{ 1 \} & \mbox{if} \ (1-p) < (1-b_i), \\
    \{ 0 \} & \mbox{if} \ (1-p) > (1-b_i), \\
    [0,1] & \mbox{otherwise}. 
    \end{cases} 
\end{align*}
\end{lemma}
\begin{proof} 
We reason for $A$. The argument for $B$ is symmetric.
Observe first of all that Equation \eqref{eq:nu} can be rewritten as 
$
    U^A_i(p,s^A_i) 
     = \frac{b_i}{p} (s^A_i (1 - p) + p) + (1-b_i) (1-s^A_i).
 $ 
So, the expected utility for strategy $s^A_i = 1$ is $\frac{b_i}{p}$ and for $s^A_i = 0$ is $1$.  If $\frac{b_i}{p}> 1$, $U^A_i(p,s^A_i) \in [1, \frac{b_i}{p}]$ and so  $s^A_1 = 1$ maximizes Equation \eqref{eq:nu}. By our assumptions, it follows that $s_i^B = 0$.
If $\frac{b_i}{p} < 1$ instead $U^A_i(p,s^A_i) \in [\frac{b_i}{p},1]$ and $s^A_1 = 0$ maximizes Equation \eqref{eq:nu}. The agent then takes the opposite side of the bet and maximizes $U^B_i(p,s^B_i)$ by setting $s_i^B = 1$. Finally, if $\frac{b_i}{p} = 1$, all investment strategies yield expected utility $1$. \qed
\end{proof}
The above result tells us that if $\s$ is in competitive equilbrium with respect to price $p(\s)$ in a Naive market, then for each agent $i$: $s^A_i = 1$ if $b_i > p(\s)$, $s^A_i = 0$ if $b_i < p(\s)$, and $s_i \in [0,1]$ if $b_i = p(\s)$. The same holds, symmetrically, for $s_i^B$.  

\medskip

Let us denote by $\linCE(\b)$ the set of investment profiles $\s$ in competitive equilibrium (under naive utilities) given beliefs $\b$. We show now that such equilibria always exist and are unique.
\begin{lemma} \label{lemma:least}
In Naive markets, for any competence profile $\q$ and belief profile $\b \in \B_\q$, $|NC(\b)| \geq 1$.
\end{lemma}
\begin{algorithm}[t] \label{algo:eq}
\SetKwInOut{Input}{input}
\SetKwInOut{Output}{output}
\caption{Competitive equilibria in Naive markets}

\Input{A belief profile $\b = (b_1, \ldots, b_n)$ ordered from highest to lowest beliefs}
\Output{An investment profile $\s = (\s^A,\s^B)$}
\BlankLine
$\s^A \gets (0, \ldots, 0)$ \tcc*[r]{We start by assuming no agent invests in $A$}
\For{$1 \leq i < n$}
{
\If{
$b_i \geq \frac{i}{n} \geq b_{i+1}$
}
{
$\s^A \gets (\underbrace{1, \ldots, 1}_{i~\mathit{times}}, 0, \ldots, 0)$ and
$\s^B \gets (\underbrace{0, \ldots, 0}_{i~\mathit{times}}, 1, \ldots, 1)$ \;
\Return{$(\s^A,\s^B)$} and exit \tcc*[r]{profile with price $\frac{i}{n}$} 
}
}
\For{$1 \leq i < n$}
{
\If{
    $\frac{i-1}{n} < b_i < \frac{i}{n}$
    }
    {
    $x \gets$ solve $\frac{1}{b_i}((i-1) + x) = \frac{1}{1-b_i}(n-i)$ \tcc*[r]{partial $A$ investment} 
    \eIf{
    $x \geq 0$
    }
   {
   $s^A_i \gets x$ \;
    $\s^A \gets (\underbrace{1, \ldots, 1}_{i-1~\mathit{times}}, s^A_i, 0, \ldots, 0)$ and
    $\s^B \gets (\underbrace{0, \ldots, 0}_{i-1~\mathit{times}}, 0, 1, \ldots, 1)$ \;
    \Return{$(\s^A,\s^B)$} and exit
    }
    {
    $x \gets$ solve $\frac{1}{b_i}(i-1) = \frac{1}{1-b_i}((n-i)+x)$ \tcc*[r]{partial $B$ investment}
    $s^B_i \gets x$ \;
    $\s^B \gets (\underbrace{0, \ldots, 0}_{i-1~\mathit{times}}, s^B_i, 1, \ldots, 1)$ and
    $\s^A \gets (\underbrace{1, \ldots, 1}_{i-1~\mathit{times}}, 0, 0, \ldots, 0)$ \;
    \Return{$(\s^A,\s^B)$} and exit \tcc*[r]{profile with price $b_i$} 
    }
}
}
\end{algorithm}
\begin{proof}
We prove the claim by construction via Algorithm \ref{algo:eq}, by showing that the algorithm outputs an investment profile which is a competitive equilibrium. 

The algorithm consists of two routines: lines 1-7, and lines 8-21. We first show that, via these two routines, the algorithm always yields an output: if the first routine does not return an output, the second one does. The two routines compare entries in two vectors:  the $n$-long vector of beliefs $(b_1, \ldots, b_n)$, assumed to be ordered by decreasing values (thus, stronger beliefs first); the $n+1$-long vector $(0,\frac{1}{n}, \frac{2}{n}, \ldots, \frac{n}{n})$, ordered therefore by increasing values. The two vectors define two functions from $\{0, \ldots, n \}$ to $[0,1]$ (we postulate $b_0 = 1$). Because the first function is non-increasing, and the second one is increasing and its image contains both $0$ and $1$, there exists $i \in \{0, \ldots n \}$ such that the two segments $[b_{i+1},b_i]$ and $[\frac{i}{n},\frac{i+1}{n}]$ intersect. There are two cases: $\frac{i}{n}$ lies in $[b_{i+1},b_i]$, in which case the condition of the first routine applies; or $b_{i+1}$ lies in $[\frac{i}{n},\frac{i+1}{n}]$, in which case the condition of the second loop applies. 

It remains to be shown that the outputs of the two routines are equilibria. The output of the first routine is an investment profile $\s = (\s^A,\s^B)$ where $i$ agents fully invest in $A$ and the remaining agents fully invest in $B$, yielding a price $p(\s) = \frac{i}{n} \in [b_i, b_{i+1}]$. By Lemma \ref{lemma:naive} such a profile is an equilibrium. The output of the second routine is an investment profile $\s$ where $i-1$ agents fully invest in $A$, $n-i$ agents fully invest in $B$ and agent $i$, whose belief equals the price, invests partially in either $A$ or $B$ in order to meet the clearing Equation \eqref{eq:balance}. By Lemma \ref{lemma:naive} we conclude that the profile is in equilibrium for $b_i$. \qed
\end{proof}
Observe that the price constructed by Algorithm \ref{algo:eq} lies in the $[b_i, b_{i+1}]$ interval.

\begin{lemma} \label{lemma:most}
Under Naive utilities, for any competence profile $\q$ and belief profile $\b \in \B_\q$, $|NC(\b)| \leq 1$.
\end{lemma}
\begin{proof}
Assume towards a contradiction there exist $\s \neq \t \in NC(\b)$. It follows that $p(\s) \neq p(\t)$. Assume w.l.o.g. that $p(\s) < p(\t)$. By Equation \ref{eq:balance} and the definition of competitive equilibrium, it follows that $\sum_{i \in N} s^A_i \leq  \sum_{i \in N} t^A_i$ (larger $A$-investment in $\t$). By Lemma \ref{lemma:naive} it follows that there are more agents $i$ such that $b_i > p(\t)$ rather than $b_i > p(\s)$, and therefore that $p(\t) < p(\s)$. A contradiction follows. \qed
\end{proof}

We can thus conclude that in Naive markets there exists exactly one competitive equilibrium and, therefore, only one equilibrium price.
\begin{theorem}
In Naive markets, for any competence profile $\q$ and belief profile $\b \in \B_\q$, $NC(\b)$ is a singleton.
\end{theorem}
\begin{proof}
The result follows directly from Lemmas \ref{lemma:least} and \ref{lemma:most}. \qed
\end{proof}
We will refer the equilibrium profile as $\s_{NC}(\b)$ and to its equilibrium price as $p_{NC}(\b)$. An interesting consequence of the above results is that such equilibrium price behaves like a quintile of $\b$, splitting the belief profile into segments roughly proportional to the price.
\begin{corollary} \label{cor:quantile}
In Naive markets, for any competence profile $\q$ and belief profile $\b \in \B_\q$, there are $n \cdot p(\s)$ agents $i$ such that $b_i \geq p_{NC}(\b)$ and there are $n \cdot  (1-p(\s))$ agents $i$ such that $b_i \leq p_{NC}(\b)$.
\end{corollary}
The equilibrium price $p_{NC}(\b)$ corresponds to the $(1-p_{NC}(\b))$-quantile of $\b$.\footnote{A similar observation, but for a continuum of players ($N = [0,1]$) and for subjective beliefs, is made in \cite{manski06interpreting}.}


\subsection{The average belief is the equilibrium price in Kelly markets} \label{sec:average}

The two following lemmas are known results from the betting \cite{kelly56new} and the information markets literature \cite{beygelzimer2012learning}, which we restate here for completeness. Proofs are also provided in Appendix \ref{app:average}. 

\begin{lemma}[\cite{kelly56new}] \label{lemma:kelly}
In Kelly markets, for any $b_i \in [0,1]$ and $p \in [0,1]$:
 \begin{align*}
\argmax_{x \in [0,1]} U_i^A(p,x) & = 
    \begin{cases}
    \frac{b_i - p}{1-p} & \mbox{if} \ p < b_i, \\
    0 & \mbox{otherwise}, 
    \end{cases} \\
\argmax_{x \in [0,1]} U_i^B(p,x) & = 
    \begin{cases}
    \frac{p - b_i}{p} & \mbox{if} \ (1-p) < (1-b_i), \\
    0 & \mbox{otherwise}. 
    \end{cases}
\end{align*}
\end{lemma}

So, a strategy profile $\s$ is in Kelly competitive equilibrium with respect to price $p(\s)$ whenever Equation \eqref{eq:balance} is satisfied together with the `Kelly conditions' of Lemma \ref{lemma:kelly}. Unlike in the case of Naive markets it is easy to see that such equilibrium is unique. So, for a given belief profile $\b$, let us denote by $\s_\logCE(\b)$ such competitive equilibrium and by $p_\logCE(\b)$ the price at such equilibrium. 

\begin{lemma}[\cite{beygelzimer2012learning}] \label{lemma:average}
For any 
$\q$ and 
$\b \in \B_\q$,
$
p_\logCE(\b) = \frac{1}{|N|}\sum_{i \in N} b_i.
$
\end{lemma}


\section{Truth-Tracking via Equilibrium Prices} \label{sec:Naive_Kelly_truth}

In this section we show how competitive equilibria in Naive and Kelly markets correspond to election by simple majority and, respectively, by a majority in which agents carry weight proportional to their competence minus $0.5$.

\subsection{Simple majority and Naive markets}

\begin{wrapfigure}{r}{4cm}
\begin{center}
\begin{tikzpicture}[node distance=2cm, auto]
    \node (1) {$\b$};
    \node (2) [below of=1] {$\bi{\b}$};
    \node (3) [right of=1] {$p_\linCE(\b)$};
    \node (4) [below of=3] {$M^\triw(\bi{\b})$};
    \draw[->] (1) to node {$\bi{}$} (2);
    \draw[->] (3) to node {$\bi{}$} (4);
    \draw[->] (1) to node {$\linCE$} (3);
    \draw[->] (2) to node {$M^\triw$} (4);
\end{tikzpicture}
\end{center}
\vspace{-0.4cm}
\caption{Simple majority and Naive markets commute.}
\vspace{-0.5cm}
\end{wrapfigure}

The following result shows that simple majority is implemented in competitive equilibrium by a Naive market: for any belief profile $\b$ induced by independent individual competences in $(0.5, 1]$, the diagram on the right commutes.
That is, the outcome of simple majority always consists of the security that the $(1-p)$-quantile belief (where $p$ is the equilibrium price) would invest in equilibrium when the market is naive. So we can treat $\linCE$ as a belief aggregator $[0,1]^n \to [0,1]$ mapping belief profiles to prices induced by competitive equilibria.

\begin{theorem}\label{th:simple_maj}
In Kelly markets, for any competence profile $\q$ and $\b \in \B_\q$:
\[
M^\triw(\bi{\b}) = \bi{p_{NC}(\b)}.
\]
\end{theorem}
\begin{proof}
The claim follows from the observation that, by Corollary \ref{cor:quantile}, $p_{NC}(\b) > 0.5$ if and only if there exists a majority of traders whose beliefs are higher than the price. From which we conclude that $\hat{\b}$ contains a majority of votes for $A$. \qed
\end{proof}

\begin{remark}
Note that, by Theorem \ref{th:simple_maj}, known extensions of the Condorcet Jury Theorem with heterogeneous competences \cite{grofman83thirteen} directly apply to Naive markets in competitive equilibrium. In particular with $N \to \infty$ the probability that $p_\linCE(\b)$ is correct approaches $1$ for any $\b$ induced by a competence profile.
\end{remark}

\subsection{Weighted majority and Kelly markets}

A similar result to Theorem \ref{th:simple_maj} can be obtained for the weighted majority rule with individual weights proportional to $q_i - 0.5$, for each individual $i$. Such a rule is implemented in competitive equilibrium by Kelly markets. Intuitively, such markets then implement a majority election where individuals' weights are proportional to how better the individual is compared to an unbiased coin.

\begin{theorem} \label{th:kelly}
In Kelly markets, for any competence profile $\q$, and $\b \in \B_\q$:
\[
M^\w(\bi{\b}) = \bi{p_\logCE(\b)},
\]
where $\w$ is such that for all $i \in N$, $w_i \propto 2q_i - 1$.
\end{theorem}

\begin{proof}
By the assumed weight profile, the normalized total weight of votes for $A$ is $\frac{\sum_{i \in N} \mathbbm{1}(b_i > 0.5) (2q_i - 1)}{\sum_{i \in N} (2q_i - 1)}$. For $A$ (respectively, $B$) to be chosen, this value should exceed (resp., fall short of) $\frac{1}{2}$. This is the case if and only if $\frac{\sum_{i \in N} (2q_i - 1)}{n} \cdot \left(\frac{\sum_{i \in N} \mathbbm{1}(b_i > 0.5)(2q_i - 1)}{\sum_{i \in N} (2q_i - 1)} - \frac{1}{2}  \right) + \frac{1}{2}$ exceeds $\frac{1}{2}$.  Let us denote this value $\rho(\b)$. The following series of equivalences shows that $\rho(\b)$ equals the average belief in $\b$.
\begin{align*}
\rho(\b)			& =  \frac{\sum_{i \in N} (2q_i - 1)}{n} \cdot \left(\frac{\sum_{i \in N} \mathbbm{1}(b_i > 0.5)(2q_i - 1)}{\sum_{i \in N} (2q_i - 1)} - \frac{1}{2}  \right) + \frac{1}{2}  \\
				& = \frac{\sum_{i \in N} (2q_i - 1)}{n} \cdot \left(\frac{2 \sum_{i \in N} \mathbbm{1}(b_i > 0.5) (2q_i - 1) - \sum_{i \in N} (2q_i - 1)}{2 \sum_{i \in N} (2q_i - 1)} \right) + \frac{1}{2}  \\
				& = \frac{\sum_{i \in N} \mathbbm{1}(b_i > 0.5) (2q_i - 1) - \sum_{i \in N} (q_i - \frac{1}{2}) + \frac{n}{2}}{n}  \\
				& = \frac{\sum_{i \in N} \mathbbm{1}(b_i > 0.5) (2q_i - 1) - \sum_{i \in N} q_i + \frac{n}{2} + \frac{n}{2}}{n}  \\
				& = \frac{\sum_{i \in N} \mathbbm{1}(b_i > 0.5) (2q_i - 1) + \sum_{i \in N} (1 - q_i)}{n}  \\
				& = \frac{\sum_{i \in N} \mathbbm{1}(b_i > 0.5) (2q_i - 1) + (1 - q_i)}{n}  \ \ \ \ \ \ \  ~\mbox{(recall Equation \eqref{eq:belief})}\\
				& = \frac{\sum_{i \in N} b_i}{n}. 
\end{align*}
From this, the definition of $M^\v(\bi{\b})$ (Equation \eqref{eq:maj}), and Lemma \ref{lemma:average} we obtain 
$
M^\v(\bi{\b}) = \bi{\frac{\sum_{i \in N} b_i}{n} } = \bi{p_\logCE(\b)},
$
as desired. \qed
\end{proof}

\begin{wrapfigure}{r}{4cm}
\vspace{-1.2cm}
\begin{center}
\begin{tikzpicture}[node distance=2cm, auto]
    \node (1) {$\b$};
    \node (2) [below of=1] {$\bi{\b}$};
    \node (3) [right of=1] {$p_\logCE(\b)$};
    \node (4) [below of=3] {$M^\w(\bi{\b})$};
    \draw[->] (1) to node {$\bi{}$} (2);
    \draw[->] (3) to node {$\bi{}$} (4);
    \draw[->] (1) to node {$\logCE$} (3);
    \draw[->] (2) to node {$M^\w$} (4);
\end{tikzpicture}
\end{center}
\vspace{-0.4cm}
\caption{ Weighted majority with weights $q_i-0.5$ and Kelly markets commute.}
\vspace{-1.5cm}
\end{wrapfigure}
Intuitively, the theorem tells us that  by implementing a weighted average of the beliefs of the traders, the competitive equilibrium price in markets with Kelly utilities behaves like a weighted majority where agents' weights are a linear function of their individual competence (specifically, $2 q_i - 1$). So, for any belief profile $\b$ induced by a competence profile $\q$ and by weights $w_i = 2q_i - 1$, we again a realization of Figure \ref{fig:commute} depicted in the commutative diagram on the right.


\section{Markets for Perfect Elections} \label{sec:perfect}

In this section we show how, by introducing a specific tax scheme, we can modify Kelly markets to make their equilibrium price implement a perfect weighted majority. Recall that we refer to perfect majority voting as weighted majority voting in which the weight of each individual is proportional to the natural logarithm of their competence ratio (recall Theorem \ref{th:grofman}). The intuition of our approach is the following: Theorem \ref{th:kelly}, has shown that Kelly markets correspond to elections where individuals are weighted proportionally to their competence in excess of $0.5$. In order to bring such weights closer to the ideal values of Theorem \ref{th:grofman} we need therefore to allow more competent agents to exert substantially more influence on the equilibrium price; we do so by designing a tax scheme which achieves such effect asymptotically in one parameter of the scheme.

\subsection{Taxing payouts}

We modify Equation \eqref{eq:kelly} by building in the effects of a tax scheme $T$ as follows:
\begin{equation} \label{eq:taxed}
U^A_i(p,s_i) = b_i \ln T\left(s_i \frac{1-p}{p} + 1 \right) + (1-b_i) \ln (1-s_i),
\end{equation}
where
\begin{equation} \label{eq:tax}
T(x) = \frac{1 - e^{-k x \frac{p}{1-p}}}{k \frac{p}{1-p}},
\end{equation}
with $k \in \mathbb{R}_{>0}$. Observe that, as parameter $k$ approaches $0$, $T(x)$ approaches $x$ and null taxation is therefore approached. 

To gain an intuition of the working of function $T$, it is useful to observe its effects on the agent's optimal investment strategy supposing the price $p = 0.5$. For $p = 0.5$ the optimal strategy of a Kelly trader is $2b_i - 1$ (Lemma \ref{lemma:kelly}). Function $T$ makes that strategy asymptotically proportional to $\ln\left(\frac{b_i}{1-b_i}\right)$ (Figure \ref{fig:tax}) as $k$ grows.

\begin{figure}[t]
\begin{center}
\includegraphics[scale=0.14]{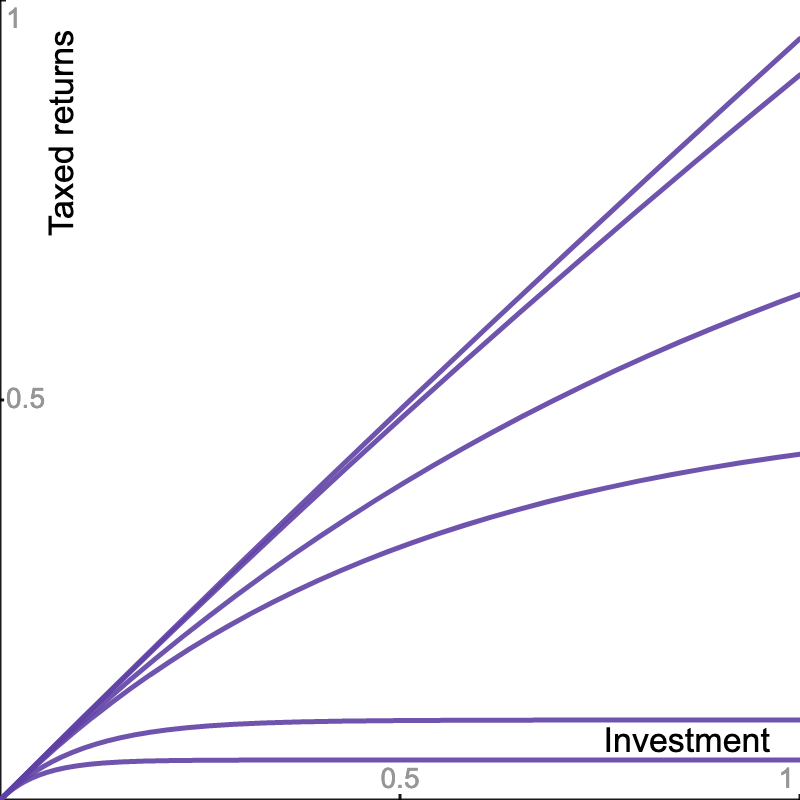} \hspace{1cm}
\includegraphics[scale=0.14]{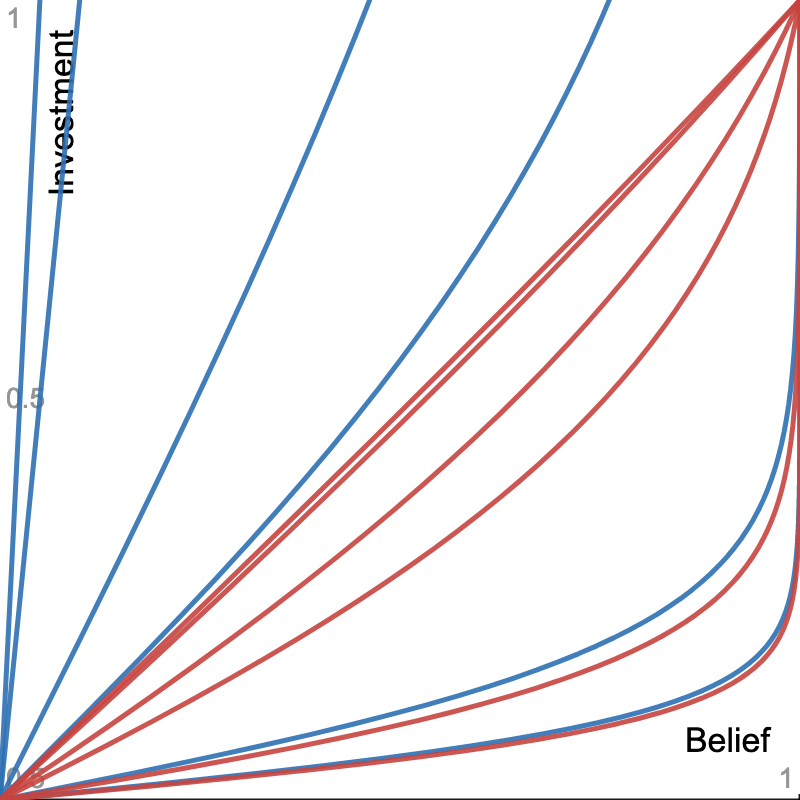}
\end{center}
\caption{Left: returns after taxation by $T$ as a function of investment (Equation \eqref{eq:taxed}). Right: investment strategy (\textcolor{red}{red}) approximating $\ln\left(\frac{b_i}{1-b_i}\right)\frac{1}{k}$ (\textcolor{blue}{blue}) as $k$ grows when price equals $0.5$. Functions plotted for $k \in \{ 0.1,0.2,1,2,10,20 \}$. }
\label{fig:tax}
\end{figure}

We call markets under the utility in Equation \eqref{eq:taxed} {\em taxed markets} and denote their equilibrium prize by $p_\taxCE(\b)$ for any belief profile $\b$.


\subsection{Equilibria in taxed Kelly markets}

Like for Naive and Kelly markets, we first determine the optimal strategy of the traders. We do that for $A$-traders, as the lemma for $B$-traders is symmetric.

\begin{lemma} \label{lemma:taxed}
In Taxed markets, for any $i \in N$, if $b_i > p$, then as $k \to \infty$, 
$$
\argmax_{x \in [0,1]} U_i^A(p,x) \propto \ln \left( \frac{1-p}{p} \cdot \frac{b_i}{1-b_i}\right).
$$
\end{lemma}
\begin{proof}
We start from $i$'s utility, given by Equation \eqref{eq:taxed}. By setting $\frac{dU_i^A}{d s_i} = 0$ (first order condition) we obtain:
\begin{equation} \label{eq:deriv}
\frac{b T'(s \frac{1-p}{p}) \frac{1-p}{p}}{1 + T(s \frac{1-p}{p})} = \frac{1-b_i}{1-s_i}
\end{equation}
If we replace Equation \eqref{eq:tax} into Equation \eqref{eq:deriv}, we obtain:
\begin{equation}
    \frac{b e^{-k s_i} \frac{1-p}{p}}{1 + \frac{1 - e^{-k s_i}}{k \frac{p}{1-p}}} = \frac{1-b_i}{1-s_i}.
\end{equation}
and therefore
\begin{equation}
    \frac{k b e^{-k s_i}}{k \frac{p}{1-p} + 1 - e^{-k s_i}} = \frac{1-b}{1-s_i}.
\end{equation}
As $k$ approaches infinity, $s_i$ approaches zero. For this reason we rescale strategies by $k$ and consider a value $y = s k$. This allows us to understand the form to which strategies tend as they approach zero. We thus obtain
\begin{equation}
    \frac{k b e^{-y}}{k \frac{p}{1-p} + 1 - e^{-y}} = \frac{1-b}{1-\frac{y}{k}}.
\end{equation}
As $k$ approaches infinity this approaches
\begin{equation}
    \frac{b e^{-y}}{\frac{p}{1-p}} = (1-b),
\end{equation}
which can be rewritten in turn as
\begin{equation}
    y = \ln \left( \frac{1-p}{p} \frac{b}{1-b}\right),
\end{equation}
from which we conclude $s_i = \frac{1}{k} \log( \frac{1-p}{p} \frac{b}{1-b})$, as desired. \qed
\end{proof}
As $k$ tends to infinity, the optimal investment strategy will tend to $0$ for all agents. However, it will do so in such a way that as $k$ grows, the optimal investment strategy tends to be proportional to $\ln(\frac{1-p}{p} \cdot \frac{b_i}{1-b_i})$ as desired. 

So, as $k$ grows large, a strategy profile $\s$ is in competitive equilibrium in a taxed market with respect to price $p(\s)$ whenever Equation \eqref{eq:balance} is satisfied together with the condition identified by Lemma \ref{lemma:taxed}. We denote by $\s_\taxCE(\b)$ such competitive equilibrium and by $p_\taxCE(\b)$ the price at such equilibrium. We then obtain the following lemma.
\begin{lemma} \label{lemma:tax_price}
In Taxed markets, for any profile $\q$ and $\b \in \B_\q$, as $k \to \infty$,
$$
\ln\left(\frac{p_\taxCE(\b)}{1-p_\taxCE(\b)} \right) 
\propto
\sum_i^n \ln \left(\frac{b_i}{1-b_i} \right).
$$
\end{lemma}
\begin{proof}
To lighten notation we write $p$ for $p_\taxCE(\b)$. From the equilibrium condition (Equation \eqref{eq:balance}) and Lemma \ref{lemma:taxed} we have that
\begin{align}
\frac{1}{p} \sum_{i \in N^A} \ln \frac{b_i}{1-b_i} & = \frac{1}{1-p} \sum_{i \in N^B}  \frac{1 - b_i}{b_i},  \label{eq:proof}
\end{align}
where $N^A = \set{i \in N \mid b_i > p}$ and $N^B = \set{i \in N \mid b_i < p}$. From the above we obtain
\begin{equation}
    0 = \sum_i^N \ln \left( \frac{1-p}{p} \frac{b_i}{1-b_i} \right),
\end{equation}
which rewrites to
\begin{equation}
    \ln \left(\frac{p}{1-p} \right) = \frac{1}{N} \sum_i^N \ln \left(\frac{b_i}{1-b_i} \right),
\end{equation}
as desired. \qed
\end{proof}
That is, the equilibrium price ratio between $A$ and $B$ securities in a taxed market tends to be proportional, in logarithmic scale, to the average belief ratio. 

\begin{theorem} \label{th:ptaxed}
In Taxed markets, for any profile $\q$, $\b \in \B_\q$ and as $k \to \infty$,
$$
M^{\w}(\bi{\b}) = \bi{p_\taxCE(\b)},
$$
where $\w$ is such that for all $i \in N$ $w_i \propto \ln \frac{q_i}{1-q_1}$.
\end{theorem}
\begin{proof}
First of all, observe that: $\bi{p_\taxCE(\b)} = \set{1}$ iff $\ln\left(\frac{p_\taxCE(\b)}{1-p_\taxCE(\b)} \right)>0$; $\bi{p_\taxCE(\b)} = \set{0,1}$ iff $\ln\left(\frac{p_\taxCE(\b)}{1-p_\taxCE(\b)} \right) = 0$; and $\bi{p_\taxCE(\b)} = \set{0}$ iff $\ln\left(\frac{p_\taxCE(\b)}{1-p_\taxCE(\b)} \right)<0$. Then, by Lemma \ref{lemma:tax_price}, Equation \eqref{eq:belief} and some algebra we obtain the following relations:
\begin{align*}
\ln\left(\frac{p_\taxCE(\b)}{1-p_\taxCE(\b)} \right) 
& \propto \sum_i^n \ln \left(\frac{b_i}{1-b_i} \right) \\
& = \sum_i^n \mathbbm{1}(b_i > 0.5) \cdot \ln \left(\frac{q_i}{1-q_i} \right)  \\
& = \sum_{i: b_i > 0.5} \ln \left(\frac{q_i}{1-q_i} \right) + \sum_{i: b_i < 0.5} \ln \left(\frac{1- q_i}{q_i} \right) \\
& = \sum_{i: b_i > 0.5} \ln \left(\frac{q_i}{1-q_i} \right) - \sum_{i: b_i < 0.5} \ln \left(\frac{q_i}{1- q_i} \right).
\end{align*}
The last expression is: positive whenever weighted voting with optimal weights returns $\set{1}$; negative whenever it returns $\set{0}$; and $0$ whenever it returns $\set{0,1}$ (Equation \eqref{eq:maj}). \qed
\end{proof}

\begin{wrapfigure}{r}{4cm}
\vspace{-1.2cm}
\begin{center}
\begin{tikzpicture}[node distance=2cm, auto]
    \node (1) {$\b$};
    \node (2) [below of=1] {$\bi{\b}$};
    \node (3) [right of=1] {$p_\taxCE(\b)$};
    \node (4) [below of=3] {$M^\w(\bi{\b})$};
    \draw[->] (1) to node {$\bi{}$} (2);
    \draw[->] (3) to node {$\bi{}$} (4);
    \draw[->] (1) to node {$\taxCE$} (3);
    \draw[->] (2) to node {$M^\w$} (4);
\end{tikzpicture}
\end{center}
\vspace{-0.4cm}
\caption{As tax parameter $k \to \infty$, perfect majority and taxed Kelly markets commute.}
\vspace{-0.5cm}
\end{wrapfigure}
This last result shows that elections that are perfect from a truth-tracking perspective (Theorem \ref{th:grofman}) can be implemented increasingly faithfully by markets with Kelly utilities, once the taxation scheme $T$ is applied and the taxation parameter $k$ in Equation \eqref{eq:tax} grows larger and, therefore, that taxation grows. So, for any belief profile $\b$ induced by a competence profile $\q$ and weights $w_i = \frac{q_i}{1-q_i}$, we obtain a realization of Figure \ref{fig:commute} consisting of the commutative diagram on the right, under the assumption that $k$ tends to infinity.


\section{Numerical examples} \label{sec:examples}

Assume $N = \set{1, \ldots, 5}$ with competence profile $\q = \tuple{0.9, 0.7, 0.6, 0.6, 0.6}$. Assume further that only the first and last agent receive signal $A$ while the rest receives signal $B$. This gives us the following belief profile by Bayesian update: $\b = \tuple{0.9, 0.3, 0.4, 0.4, 0.6}$.  These beliefs result in the voting profile $\v = \tuple{1, 0, 0, 0, 1}$, from which we obtain: 
\begin{itemize}

\item $M^\triw(\v) = \set{0}$, that is, standard majority selects $B$;

\item $M^\w(\v) = \set{1}$, where $\w = \tuple{0.8, 0.4, 0.2, 0.2, 0.2}$ (weight profile given by $2 q_i - 1$), as $0.8 + 0.2 - (0.4 + 0.2 + 0.2) > 0$, that is, the sum of weights of the first and last agents are larger then the sum of weights of the others;

\item $M^\w(\v) = \set{1}$, where $\w = \tuple{\ln \frac{0.9}{0.1}, \ln \frac{0.7}{0.3}, \ln \frac{0.6}{0.4}, \ln \frac{0.6}{0.4}, \ln \frac{0.6}{0.4}}$ (optimal weights), as the following expression is positive:
\begin{equation} \label{eq:example_maj}
\ln \frac{0.9}{0.1} + \ln \frac{0.6}{0.4} - \left(\ln \frac{0.7}{0.3}  + 2 \cdot \ln \frac{0.6}{0.4} \right).
\end{equation}
\end{itemize}

We move now to the choices made by the markets based on equilibrium prices. We have that: by Algorithm \ref{algo:eq}, $p_\linCE(\b) = \frac{2}{5}$ (Naive market equilibrium) where the two agents who received the $A$ signal invest all their endowment in $A$-securities, and the remaining agents invest all their endowment in $B$-securities; $p_\logCE(\b) = \frac{2.6}{5}$ (Kelly market equilibrium) corresponding to the mean belief in $\b$. So, a Naive market given the above beliefs selects $B$ while the Kelly market selects $A$ by a very small margin. As to the taxed markets, our results do not give us a closed expression for $p_\taxCE(\b)$ but rather determine whether the price favors $A$- or $B$-securities based on the logarithm of the ratio between the two prices, which is proportional to the logarithm of the weighed support for $A$ and for $B$ when the taxation parameter $k$ tends to infinity (Theorem \ref{th:ptaxed}). In this example, we thus have that $\ln \frac{p_\taxCE(\b)}{1 - p_\taxCE(\b)}$ is proportional to Equation \eqref{eq:example_maj} and therefore points to security $A$.

\medskip

Assume $N = \set{1, \ldots, 4}$ with competence profile $\q = \tuple{0.8, 0.6, 0.6, 0.6}$ and that only the first agent receives signal $A$ while the rest receives signal $B$. This gives us the following belief profile: $\b = \tuple{0.8, 0.4, 0.4, 0.4}$. These beliefs result in the voting profile $\v = \tuple{1, 0, 0, 0}$, from which we obtain: 
\begin{itemize}

\item $M^\triw(\v) = \set{0}$, that is, standard majority selects $B$

\item $M^\w(\v) = \set{0,1}$ where $\w = \tuple{0.6, 0.2, 0.2, 0.2}$ (weight profile given by $2 q_i - 1$) as $0.6  - (0.2 + 0.2 + 0.2) =  0$. That is, we have a split weighted majority. 

\item $M^\w(\v) = \set{1}$ where $\w = \tuple{\ln \frac{0.8}{0.2}, \ln \frac{0.6}{0.4}, \ln \frac{0.6}{0.4}, \ln \frac{0.6}{0.4}}$ (optimal weights) as
\begin{equation} \label{eq:example_maj2}
\ln \frac{0.8}{0.2} -  3 \cdot \ln \frac{0.6}{0.4}
\end{equation}
is positive.
\end{itemize}

As to equilibrium prices, by applying Algorithm \ref{algo:eq}, we have that also in this case $p_\linCE(\b) = \frac{2}{5}$ (Naive market equilibrium). This price equals the posterior beliefs of the three agents that receive signal $B$. By the algorithm, the agent receiving signal $A$ invests all its wealth in $A$, one of the agents receiving signal $B$ invests $\frac{1}{3}$ of their wealth in $A$ (to guarantee market clearing at that price, line 10 of Algorithm \ref{algo:eq}), and the remaining agents invest all their endowment in $B$-securities. The equilibrium price in Kelly markets is in this case $0.5$ (mean belief). So, a Naive market given the above beliefs selects $B$ while the Kelly market remains undecided. In taxed markets, we have that $\ln \frac{p_\taxCE(\b)}{1 - p_\taxCE(\b)}$ is proportional to Equation \eqref{eq:example_maj2} and therefore points to security $A$.


\section{Discussion and Outlook} \label{sec:conc}

Our paper is the first one to establish a formal link between voting and information markets from an epistemic social choice perspective. The link consists specifically of correspondence results between weighted majority voting on the one hand, and information markets under three types of utility on the other. Such results open up the possibility, in principle, to implement weighted majority voting with strong epistemic guarantees even without having access to individual competences, because such information becomes indirectly available in the market via the equilibrium price. Notice that, in particular, while it may be difficult to elicit truthful weights from agents, investment strategies are subject to the natural incentive of maximizing investment returns. Whether this can prove advantageous also in practice, for instance in the setting of classification markets \cite{barbu12introduction} or voting-based ensembles \cite{cornelio22voting}, should be object of future research.

\medskip

The study we presented is subject to at least four main limitations. First, our analysis inherits all assumptions built into standard jury theorems, in particular: jurors' independence; homogeneous priors; equivalence of type-$1$ and type-$2$ errors in jurors' competences; binary events. Future research should try to lift our correspondence to more general settings relaxing the above assumptions (see \cite{dietrich21jury} for a recent overview, and \cite{pivato13voting,pivato17epistemic} for more general frameworks for epistemic social choice).
Second, our study limited itself to one-shot interactions. However, markets and specifically Kelly betting make most sense in a context of iterated decisions. Extending our results to the iterated setting, along the lines followed for instance in \cite{beygelzimer2012learning}, is also a natural avenue for future research. Third, our market model makes use of the notion of competitive equilibrium. Although such notion of equilibrium is standard in information markets, it responds to the intuition that individuals operate in a large group and, therefore, behave as price takers. We consider it interesting to study how different notions of equilibrium that do not make such assumption (e.g., Nash equilibrium), would behave within our framework. Fourth, our analysis assumes very stylized utility functions which are identical for all agents. This is of course highly restrictive and our results would be substantially strengthened if lifted to more general, and possibly heterogeneous, classes of utilities.


\begin{credits}
\subsubsection{\ackname}
This research was (partially) funded by the Hybrid Intelligence Center, a 10-year programme
funded by the Dutch Ministry of Education, Culture and Science through the Netherlands
Organisation for Scientific Research, \url{https://hybrid-intelligence-centre.nl}, grant number 024.004.022. Davide Grossi wishes to also thank Universit\'{e} Paris Dauphine and the Netherlands Institute for Advanced Studies (NIAS), where parts of this research were completed. We are grateful to Marcus Pivato, Feline Lindeboom and the anonymous reviewers of COMSOC'23 and SAGT'24 for many useful comments on earlier versions of the paper.
\end{credits}



\appendix


\section{Auxiliary result: full investment constraint is without loss of generality} \label{sec:supp}

Let us define the investment strategy of a trader $i$ by the pair $s_i = (s_i^A, s_i^B)$ where $s_i^A \in [0,1]$ is the amount of endowment invested in $A$-securities (similarly for $s_i^B$). Let then $c_i = w_i - (s_i^A + s_i^B)$, that is, the unspent endowment of $i$ given strategy $(s_i^A, s_i^B)$. The assumption that traders invest their full endowment amounts to $s_i^A + s_i^B = 1$. We can therefore simply refer to $s^A_i$ as $s$ and to $s^B_i$ as $1-s$. 

In Naive markets, it is a corollary of Lemma \ref{lemma:naive} that, in equilibrium, any strategy in which traders invest all their endowment is equivalent to a strategy investing in at most one security. 

For Kelly markets we can prove that the same observation holds. Under a full investment assumption Equation \eqref{eq:kelly} needs to be modified to:
\begin{equation}
U_i(p,s) = b_i \ln \left(\frac{s}{p} \right) + (1-b_i) \ln \left(\frac{1 - s}{1-p} \right). \label{eq:kelly_full}
\end{equation} 
which can be further generalized to 
\begin{equation}
U_i(p,s) = b_i \ln \left(\frac{s}{p} + c \right) + (1-b_i) \ln \left(\frac{1 - s}{1-p} + c \right). \label{eq:kelly_free}
\end{equation} 
when $i$ can invest any amount of endowment in $A$- or $B$-securities and still keep unspent endowment $c$.

\begin{proposition} \label{prop:wlog}
In Kelly markets, for all $i \in N$, $p \in [0,1]$ and $s_i$ such that $s_i^A + s_i^B = 1$ (whole endowment is spent) there exists $t_i = (t^A_i, 0)$ or $t_i = (0, t_i^B)$ which yield the same utility as $s_i$, and vice versa.
\end{proposition}
\begin{proof}
First observe that the first-order condition for Equation \eqref{eq:kelly_full} gives us the optimal investment strategy, which is $b_i$. Such optimal strategy will give $i$ $\frac{b_i}{p}$ $A$-securities and $\frac{1- b_i}{1-p}$ $B$-securities. 
\fbox{Left to right} Suppose that $b_i > p$. It follows that $\frac{b_i}{p} > \frac{1- b_i}{1-p}$, that is, $i$ purchases more $A$-securities than $B$-securities. Were $i$ not be compelled to invest all her endowment, she would invest according to Equation \eqref{eq:kelly_free}. Now it is easy to see that if $i$ invests $b_i$ in $A$-securities, $0$ in $B$-securities, and keeps $1-b_i$ as cash, she obtains the same payoff by Equation \eqref{eq:kelly_free} as she would obtain by Equation \eqref{eq:kelly_full} when she invests $1-b_i$ in $B$-securities. There exists therefore $t_i = (t^A_i, 0)$ yielding the same utility. The same reasoning applies, symmetrically, for $b_i < p$. \fbox{Right to left}. Given an investment strategy $t_i$ investing in only one security, a utility-equivalent full-investment strategy can be constructed using the same reasoning used in the previous case via the utility for unrestricted investments of Equation \eqref{eq:kelly_free}.
\end{proof}


\section{Proofs of Section \ref{sec:average}} \label{app:average}

\begin{proof}[Proof of Lemma \ref{lemma:kelly}]
We reason for $A$. The argument for $B$ is symmetric.
First we compute the first derivative of $U^A_i(p,x)$:
\begin{align*}
\frac{d}{d x} U^A_i(p,x) & = \frac{d}{d x} \left( b_i \ln (x \frac{1-p}{p} + 1) + (1-b_i) \ln (1-x) \right) \\
				    & = b_i \cdot \frac{\frac{d}{d x}(x \frac{1-p}{p} + 1)}{x \frac{1-p}{p} + 1} + (1-b_i) \cdot \frac{\frac{d}{d x}(1-x)}{1-x} \\
				    & = b_i \cdot \frac{\frac{1-p}{p}}{x \frac{1-p}{p} + 1} + (1-b_i) \cdot \frac{-1}{1-x} \\
				    & = \frac{b_i (1-p)}{x+p(1-x)}-\frac{1-b_i}{1-x}.
\end{align*}
The first-order condition is $\frac{b_i (1-p)}{x+p(1-x)}-\frac{1-b_i}{1-x} = 0$, that is, $\frac{b_i-p}{1-b} = \frac{x}{1-x}$. Solving for $x$ we obtain $\frac{b_i-p}{1-b_i}$. We need then to check the second-order condition. The second derivative is $- \frac{1-b_i}{(1-x)^2} - \frac{b_i(1-p)^2}{(p(1-x)+x)^2}$, which is negative at every point and therefore, as desired, also at $\frac{b_i-p}{1-b_i}$. Now observe that $\frac{b_i-p}{1-b_i}$ is positive whenever $b_i > p$ and equals $0$ when $b_i = p$.
When $b_i < p$ then the expression is negative, and the trader takes the opposite side of the bet, investing in $B$. 
\end{proof}

\begin{proof}[Proof of Lemma \ref{lemma:average}]
Let us denote $p_\logCE(\b)$ by $p$. By Lemma \ref{lemma:kelly}, Equation \eqref{eq:balance} becomes:
\begin{align}
\frac{1}{p} \sum_{i \in N^A}  \frac{b_i - p}{1-p} & = \frac{1}{1-p} \sum_{i \in N^B}  \frac{p - b_i}{p}  \label{eq:proof}
\end{align}
where $N^A = \set{i \in N \mid b_i > p}$ and $N^B = \set{i \in N \mid b_i < p}$. Notice we can assume $N = N^A \cup N^B$ as any agent with belief equal to the price does not invest in either $A$- or $B$-securities in Kelly markets (by Lemma \ref{lemma:kelly}). From Equation \eqref{eq:proof}, by some basic algebra, we obtain the following series of equations
\begin{align*}
\frac{1-p}{p} \sum_{i \in N^A}  \frac{b_i - p}{1-p} & =  \sum_{i \in N^B}  \frac{p - b_i}{p} \\
\sum_{i \in N^A}  b_i - p & =  \sum_{i \in N^B}  p - b_i \\
n p & = \sum_{i \in N} b_i \\
p & = \frac{1}{n} \sum_{i \in N} b_i,
\end{align*}
thereby proving the claim.
\end{proof}



\bibliographystyle{plain}


\end{document}